\documentclass[a4paper,UKenglish,cleveref, autoref, thm-restate]{article}
\usepackage[hmargin=1.2in, vmargin=1in]{geometry}
\usepackage{booktabs}
\usepackage{tikz}
\usepackage{amsmath}
\usepackage{amssymb}
\usepackage{amsthm}
\usepackage{verbatim}
\usepackage{bbm}
\usepackage{hyperref}
\usepackage{thmtools}
\usepackage{enumerate}   
\usepackage[normalem]{ulem}
\declaretheorem{theorem}
\declaretheorem[sibling=theorem]{lemma}
\declaretheorem[sibling=theorem]{proposition}
\declaretheorem[sibling=theorem]{corollary}
\newcommand{\floor}[1]{\lfloor #1 \rfloor}
\newcommand{\bigfloor}[1]{\bigg\lfloor #1 \bigg\rfloor}

\usepackage{lineno}

\title{Enumeration of rooted binary perfect phylogenies} 

\author{Chloe E.~Shiff\thanks{Institute for Computational and Mathematical Engineering, Stanford University, Stanford, CA 94305, USA. cshiff@stanford.edu.}, Noah A.~Rosenberg\thanks{Department of Biology, Stanford University, Stanford, CA 94305, USA. noahr@stanford.edu.}}

\begin{document}
\maketitle

\begin{abstract}
Rooted binary perfect phylogenies provide a generalization of rooted binary unlabeled trees in which each leaf is assigned a positive integer value that corresponds in a biological setting to the count of the number of indistinguishable lineages associated with the leaf. For the rooted binary unlabeled trees, these integers equal 1. We enumerate rooted binary perfect phylogenies with $n \geq 1$ leaves and sample size $s$, $s \geq n$: the rooted binary unlabeled trees with $n$ leaves in which a sample of size $s \geq n$ lineages is distributed across the $n$ leaves. 
(1) First, we recursively enumerate rooted binary perfect phylogenies with sample size $s$, summing over all possible $n$, $1 \leq n \leq s$. We obtain an equation for the generating function, showing that asymptotically, the number of rooted binary perfect phylogenies with sample size $s$ grows with $\approx 0.3519(3.2599)^s s^{-3/2}$, faster than the rooted binary unlabeled trees, which grow with $\approx 0.3188(2.4833)^s s^{-3/2}$.
(2) Next, we recursively enumerate rooted binary perfect phylogenies with a specific number of leaves $n$ and sample size $s \geq n$. We report closed-form counts of the rooted binary perfect phylogenies with sample size $s \geq n$ and $n=2$, 3, and 4 leaves. We provide a recurrence for the generating function describing, for each number of leaves $n$, the number of rooted binary perfect phylogenies with $n$ leaves as the sample size $s$ increases. 
(3) Finally, we recursively enumerate rooted binary perfect phylogenies with a specific unlabeled tree shape and sample size $s$. In a special case, we find a generating function for the number of rooted binary perfect phylogenies with the $n$-leaf caterpillar shape, growing with $s$. We also find a generating function for the number of rooted binary perfect phylogenies with sample size $s$ and any caterpillar tree shape.
The enumerations further characterize the rooted binary perfect phylogenies, which include the rooted binary unlabeled trees, and which can provide a set of structures useful for various biological contexts. 
\end{abstract}


\section{Introduction}\label{sec:intro}

Rooted binary unlabeled tree structures are classic objects of combinatorics and graph theory~\cite{drmota, flajolet_sedgewick_2009}. In evolutionary biology, rooted binary unlabeled trees are used to describe the possible relationships that a set of unlabeled organisms can possess, so that they arise in inferences about features of speciation histories~\cite{Felsenstein05, Steel16}.

The rooted binary unlabeled trees can be enumerated recursively. Denoting by $u_n$ the number of rooted binary unlabeled trees with $n$ leaves, for $n \geq 2$, the recursion is
\begin{equation}
\label{eq:wedderburn_rec}
u_n=
\begin{cases}
\sum\limits_{i=1}^{n-1} \frac{1}{2} u_{n-i}u_i, & \text{odd } n \geq 3, \\
\bigg(\sum\limits_{i=1}^{n-1} \frac{1}{2} u_{n-i}u_i \bigg) +\frac{1}{2}u_{n/2}, & \text{even } n \geq 2,
\end{cases} 
\end{equation}
with $u_0=0$ and $u_1=1$. The recursion is obtained by summing over possible numbers of leaves $i$ for the right-hand subtree descended from the root. The factor of $\frac{1}{2}$ arises from the fact that each tree is obtained twice---once with its left and right subtrees transposed. If $n$ is even, for $i=\frac{n}{2}$, the recurrence counts the $\binom{u_{n/2}}{2}$ trees with distinct subtrees and the $u_{n/2}$ trees with identical subtrees. The $u_n$, $n \geq 1$, follow the Wedderburn-Etherington sequence (OEIS A001190), with initial terms $1, 1, 1, 2, 3, 6, 11, 23, 46, 98$; they have generating function
\begin{equation}
\label{eq:U}
U(z) = \frac{1}{2} U(z)^2 + \frac{1}{2} U(z^2) + z.
\end{equation}
The convergence radius for $U(z)$ is approximately 0.4027. The asymptotic growth of $u_n$ approximately follows $0.3188(2.4833)^{n}n^{-3/2}$~\cite{Comtet74, Harding71, Landau77}.

Rooted binary \emph{perfect phylogenies} can be viewed as generalizing the rooted binary unlabeled trees. A rooted binary perfect phylogeny is a rooted binary tree in which each leaf is associated with a positive integer~\cite{palacios2022}. Each integer can be regarded as a multiplicity for the biological entity associated with a leaf---for example, the number of times that a specific DNA sequence is seen in a sample of sequences that are not necessarily distinct. A rooted binary perfect phylogeny has a number of leaves $n$ and a sample size $s$ that represents the sum of the multiplicities at the leaves. Rooted binary perfect phylogenies are a special case of rooted \emph{multifurcating} perfect phylogenies---perfect phylogenies in which internal nodes possess two or more immediate descendants~\cite{palacios2022}. The rooted binary unlabeled trees correspond to rooted binary perfect phylogenies in the case that $s=n$ (and all leaf multiplicities equal 1).

In evolutionary biology, perfect phylogenies can sometimes be used as representations of the relationships of genetic sequences~\cite{Gusfield91, palacios2022}. The topology of a perfect phylogeny encodes ancestral relationships in a set of sequences that have not experienced recombination, and in which each mutation has occurred only once. \emph{Rooted} perfect phylogenies have one vertex designated as the root, representing a sequence from which all other sequences in the perfect phylogeny descend.

Palacios et al.~\cite{palacios2022} have recently developed the enumerative combinatorics of rooted perfect phylogenies, focusing on enumerations of various classes of binary trees that are \emph{compatible} with a given rooted perfect phylogeny. In the work of Palacios et al.~\cite{palacios2022}, all possible rooted binary trees (with leaf multiplicities equal to 1) that can be ``collapsed'' into a specific perfect phylogeny (with leaf multiplicities possibly greater than 1) are enumerated. To facilitate the enumerations, Palacios et al.~\cite{palacios2022} defined a partial order on rooted binary perfect phylogenies, inducing a lattice structure for the rooted binary perfect phylogenies with fixed sample size $s$. The lattice structure defines the sense in which trees can be ``collapsed.''

The lattice of Palacios et al.~\cite{palacios2022} describes a space of possible rooted binary perfect phylogenies; although Palacios et al.~\cite{palacios2022} focused on using the lattice to enumerate rooted binary trees associated with a rooted binary perfect phylogeny, the lattice formulation provides a useful structure for understanding the perfect phylogenies themselves. In Section~\ref{sec:s_only}, we recursively enumerate the rooted binary perfect phylogenies with sample size $s$, considering all possible values of the number of leaves $n$; we provide the asymptotic approximation of this quantity as $s$ increases. In Section~\ref{sec:s,n}, we recursively enumerate rooted binary perfect phylogenies with a specific number of leaves $n$ and sample size $s$. We provide a recursive equation to compute, for each $n$, the generating function for the sequence of the number of rooted binary perfect phylogenies with the number of leaves $n$ fixed and the sample size $s$ growing.  In Section~\ref{sec:treeshape}, we recursively enumerate the rooted binary perfect phylogenies with sample size $s$ for a fixed unlabeled tree shape. When this tree shape is a caterpillar, we obtain, for each small $n$, a closed-form expression for the number of perfect phylogenies for any $s$. We also provide generating functions for the number of perfect phylogenies with an $n$-leaf caterpillar shape and with any caterpillar shape, growing with $s$.

\section{Preliminaries}\label{sec:prelims}

\subsection{Definitions}

We restrict attention to rooted perfect phylogenies that are \emph{binary}, with each internal node possessing exactly two child nodes. Henceforth, the perfect phylogenies that we consider are understood to be rooted and binary, and we sometimes omit these descriptors. Like the rooted binary unlabeled trees, we consider perfect phylogenies to be \emph{non-plane} trees, so that the left--right order in which child nodes are depicted is ignored.

We denote the \emph{number of leaves} in a perfect phylogeny by $n$. Each leaf is associated with a positive integer, its \emph{multiplicity}---representing in biological applications of perfect phylogenies the number of copies of a biological sequence seen in a sample of sequences. We refer to the sum of the multiplicities in a perfect phylogeny as its \emph{sample size}.

For fixed sample size $s$, it is convenient for the lattice formulation to allow an \emph{empty} perfect phylogeny, though we exclude this empty perfect phylogeny from our enumerations. We do include the \emph{trivial} perfect phylogeny with sample size $s$, namely a perfect phylogeny that consists only of a single leaf of multiplicity $s$.

Figure~\ref{fig:1} displays an example perfect phylogeny with $n=8$ leaves. The leaf multiplicities are 2, 1, 3, 3, 1, 4, 1, and 2, for a total sample size $s=17$.
\begin{figure}
\centering
\includegraphics[width=4cm]{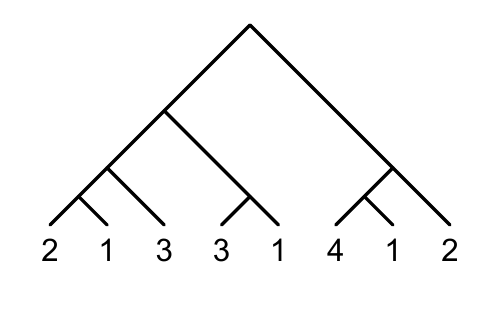}
\vspace{-.5cm}
\caption{A perfect phylogeny with sample size $s=17$ and $n=8$ leaves. The numbers at the leaves represent leaf multiplicities.}
\label{fig:1}
\end{figure}

\subsection{Lattices of perfect phylogenies}

Palacios et al.~\cite{palacios2022} defined a partial order on the rooted binary perfect phylogenies with sample size $s$. Recall that a \emph{cherry node} in a rooted tree is an internal node with precisely two descendant leaves. Consider binary perfect phylogenies $A$ and $B$. In the partial order, $A$ \emph{refines} $B$ if by collapsing cherries of $A$, $B$ can be produced. Trivially, a perfect phylogeny refines itself. $A$ and $B$ are \emph{comparable} if $A$ refines $B$ or $B$ refines $A$.

Considering all binary perfect phylogenies with sample size $s$, the partial order of Palacios et al.~\cite{palacios2022} produces a lattice. Figure~\ref{fig:2lattice5} depicts the lattice for the case of $s=5$. Moving left to right, a path is drawn between pairs $(A,B)$, with $A$ to the left of $B$, if and only if $A$ refines $B$. The trivial perfect phylogeny of sample size $s$ is refined by all perfect phylogenies of sample size $s$ and is the maximal element of the lattice. The empty perfect phylogeny refines all perfect phylogenies of sample size $s$ and is the minimal element.

\subsection{Description of the enumeration problems}
\label{sec:classes}

We enumerate several sets of objects. First, we consider the set of (non-empty) rooted binary perfect phylogenies with fixed sample size $s \geq 1$. Denote the size of this set by $b_s$. Next, we enumerate the rooted binary perfect phylogenies with fixed sample size $s$ and fixed number of leaves $n$, $1 \leq n \leq s$. Denote the size of this set by $b_{s,n}$; we have $b_s = \sum_{n=1}^s b_{s,n}$. We enumerate the rooted binary perfect phylogenies with sample size $s$ and a caterpillar topology with $n$ leaves, where $s \geq n \geq 2$, denoting this quantity $g_{s,n}$; we also enumerate the rooted binary phylogenies with sample size $s$ and \emph{any} caterpillar topology, denoting this quantity $g_s = \sum_{n=2}^s g_{s,n}$. Finally, we enumerate the rooted binary perfect phylogenies with sample size $s$ and a specified unlabeled topology $T$, denoting this quantity $N_{s,T}$. 

Note that we have already described $b_{n,n}$ (alternatively, $b_{s,s}$), the number of rooted binary perfect phylogenies with sample size equal to the number of leaves in eq.~\eqref{eq:wedderburn_rec}. A rooted binary perfect phylogeny with $s=n$ is simply a rooted binary unlabeled tree; each leaf multiplicity in the perfect phylogeny is 1, so that the rooted binary unlabeled trees correspond to the rooted binary perfect phylogenies in which all leaf multiplicities equal 1. Hence, $b_{n,n}=u_n$.

\section{Rooted binary perfect phylogenies with sample size \texorpdfstring{$s$}{s}}\label{sec:s_only}

\subsection{Enumeration}

To enumerate all rooted binary perfect phylogenies with a fixed sample size $s \geq 1$, we first note that for each $s$, the trivial perfect phylogeny is permissible. If a perfect phylogeny is not trivial, then each of the two child nodes of the root is itself the root of a perfect phylogeny. In other words, the perfect phylogeny can be decomposed into two perfect phylogenies, one with sample size $i$, $1 < i < s$, and the other with sample size $s-i$.

\begin{figure}[tb]
\centering
\includegraphics[width=0.94\textwidth]{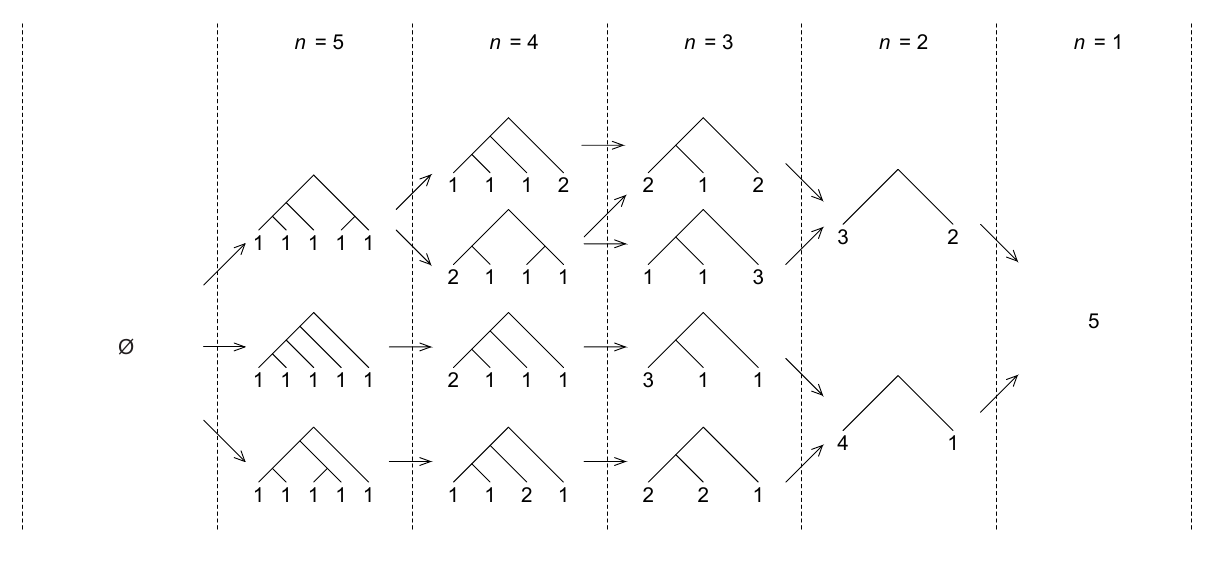}
\vspace{-.5cm}
\caption{The lattice of rooted binary perfect phylogenies for sample size $s=5$. Each column is labeled by its associated number of leaves $n$.} 
\label{fig:2lattice5}
\end{figure}

Noting that for $i=s-i$, we count the $\binom{b_{s/2}}{2}=b_{s/2}(b_{s/2}-1)/2$ perfect phylogenies with distinct perfect phylogenies in the two children of the root and the $b_{s/2}$ perfect phylogenies with identical perfect phylogenies in the subtrees, and noting that $\binom{b_{s/2}}{2}+b_{s/2}=\frac{1}{2}(b_{s/2}^2+b_{s/2})$. We obtain the following result.
\begin{proposition}
The number $b_s$ of rooted binary perfect phylogenies with sample size $s\geq 2$ satisfies
\begin{equation}
b_s=\begin{cases}
1+ \sum\limits_{i=1}^{s-1} \frac{1}{2}b_{s-i}b_i, & \text{odd } s \geq 3, \\
1+\bigg( \sum\limits_{i=1}^{s-1} \frac{1}{2} b_{s-i}b_i \bigg) + \frac{1}{2} b_{s/2}, & \text{even } s \geq 2,\\
\end{cases}
\label{eq:enumeration_alln}
\end{equation}
with $b_0=0$ and $b_1=1$.
\label{prop:bs recursion}
\end{proposition}
The recursion has the same form as eq.~\eqref{eq:wedderburn_rec}, adding a $+1$ term for the trivial perfect phylogeny. The first terms of the sequence appear in Table~\ref{table:perfect}, along with the Wedderburn-Etherington numbers of rooted binary unlabeled trees. The number of rooted binary perfect phylogenies $b_s$ with sample size $s$ appears to grow substantially faster than $b_{s,s}$, the number of rooted binary unlabeled trees with sample size $s$ and multiplicity 1 assigned to each leaf.

\subsection{Generating function}

To analyze the asymptotic growth of the rooted binary perfect phylogenies with sample size $s$ as $s \rightarrow \infty$, we rewrite eq.~\eqref{eq:enumeration_alln} in the form
\begin{equation}
b_s=\frac{1}{2} \bigg( \sum_{i=1}^{s-1} b_{s-i}b_i \bigg) +\frac{1}{2} b_{s/2}+1, \quad s\geq 1,
\label{eq:recursion}
\end{equation}
with base case $b_0=0$ and $b_s=0$ if $s$ is not a positive integer.

Denote by $B(z)$ the generating function for the rooted binary perfect phylogenies with sample size $s$, $B(z)=\sum_{s=0}^\infty b_sz^s$. To obtain the generating function for the $b_s$, we multiply eq.~\eqref{eq:recursion} by $z^s$ and sum from $s=0$ to $\infty$, obtaining
\begin{eqnarray}
B(z) &=& \sum_{s=1}^\infty \bigg( \frac{1}{2}\sum_{i=1}^{s-1} b_{s-i}b_i z^s \bigg) +\sum_{s=1}^\infty \frac{1}{2} b_{s/2} z^s+\sum_{s=1}^\infty z^s.
\nonumber
\end{eqnarray}
We simplify by noting $\sum_{s=1}^\infty ( \frac{1}{2}\sum_{i=1}^{s-1} b_{s-i}b_i z^s ) = \frac{1}{2}B(z)^2$, $\sum_{s=1}^\infty \frac{1}{2} b_{s/2} z^s = \frac{1}{2} \sum_{s=1}^\infty b_s z^{2s} = \frac{1}{2} B(z^2)$, and $\sum_{s=1}^{\infty} z^s = \frac{z}{1-z}$. We have therefore demonstrated the following proposition.
\begin{proposition}
\label{prop:Bs gen func}
The generating function $B(z)$ for the number $b_s$ of rooted binary perfect phylogenies with sample size $s$ satisfies
\begin{equation}
B(z)=\frac{1}{2}B^2(z)+\frac{1}{2} B(z^2)+\frac{z}{1-z}.
\label{eq:final_generating}
\end{equation}
\end{proposition}

\subsection{Asymptotics}

Now that we have obtained an equation satisfied by the generating function for the coefficients $b_s$, we find an asymptotic approximation for the growth of the $b_s$.

\begin{table}[tb]
\centering
\begin{tabular}{ccccccccccccccc}
\toprule
 $s$       & 1 & 2 & 3 & 4 &  5 &  6 &  7 &   8 &   9 & 10 & 11 & 12 & 13 & 14 \\
\midrule
 $b_s$     & 1 & 2 & 3 & 7 & 14 & 35 & 85 & 226 & 600 & 1658 & 4622 & 13141 & 37699 & 109419 \\
 $b_{s,s}$ & 1 & 1 & 1 & 2 &  3 &  6 & 11 &  23 &  46 &  98  &  207 &   451 &   983 &   2179 \\
 \bottomrule
\end{tabular}
\caption{The number of rooted binary perfect phylogenies $b_s$ with sample size $s$ (eq.~\eqref{eq:enumeration_alln}, OEIS A113822), and the number of rooted binary unlabeled trees $b_{s,s}=u_s$ (eq.~\eqref{eq:wedderburn_rec}, OEIS A001190).}
\label{table:perfect}
\end{table}

Recall that the generating function $U(z)$ for the number of rooted binary unlabeled trees in eq.~\eqref{eq:U} has a radius of convergence $\rho \approx 0.4027$. The rooted binary perfect phylogenies with sample size $s$ include the rooted binary unlabeled trees with $s$ leaves, so that $b_s \geq u_s$, and indeed $b_s > u_s$ for $s \geq 2$. Hence, we have $B(z) > U(z)$ for all $z$ with $0 < z < \rho$. Labeling the radius of convergence of $B(z)$ by $\beta$, it follows that $0 \leq \beta \leq \rho < 1$. In addition, because $z^2 < z$ for $0 < z < 1$ and $B(z)$ is monotonically increasing with $z$ for $z>0$, $B(z^2) < B(z)$ for $0 < z < 1$, so that if $B(z)$ converges at $z$, $0 < z < 1$, then it also converges at $z^2$.

To obtain the asymptotic approximation, we first prove a lemma about the relationship of the $b_s$ to the Catalan numbers.
\begin{lemma}
Each $b_s$ for $s \geq 1$ is bounded above by the Catalan number $C_s = \frac{1}{s+1} \binom{2s}{s}$.
\label{lem:bncn}
\end{lemma}

\begin{proof}

To prove $b_s\leq C_s$ for all $s\geq 1$, we recall the recursion for the Catalan numbers, $C_{s}=\sum_{k=0}^{s-1} C_k C_{s-1-k} = \sum_{k=1}^s C_{k-1} C_{s-k}$ with $C_0=1$ \cite[p.~26]{drmota}. We first prove inductively that $\frac{1}{2}b_{s+1}\leq C_s$ for all $s \geq 0$.

For the base case $s=0$, we have $\frac{1}{2} = \frac{1}{2}b_1 \leq C_0 = 1$; for $s=1$, we have $1 = \frac{1}{2}b_2 \leq C_1 = 1$. For the inductive step, suppose $\frac{1}{2}b_{s+1}\leq C_{s}$ for all $s < N$, $N \geq 2$. By the recursion in eq.~\eqref{eq:recursion} and the inductive assumption,
\begin{eqnarray}
\frac{1}{2}b_{N+1} & = & \bigg(\sum_{i=2}^{N} \frac{1}{2} b_{N+1-i}\frac{1}{2}b_i \bigg) + \frac{1}{4} b_{N} b_1 + \frac{1}{4} b_{(N+1)/2} + \frac{1}{2} \nonumber \\ 
&\leq & \bigg(\sum_{i=2}^N C_{N-i}C_{i-1} \bigg) + \frac{1}{4} b_{N} b_1 
+ \frac{1}{4}b_{(N+1)/2} + \frac{1}{2}. \nonumber 
\end{eqnarray}

By the inductive assumption, noting $b_1 = C_0 = 1$, $\frac{1}{4}b_N b_1 \leq \frac{1}{2} C_{N-1} C_0$. Also by the inductive assumption, $\frac{1}{4}b_{(N+1)/2} + \frac{1}{2} \leq \frac{1}{2}C_{(N-1)/2} + \frac{1}{2}$. The Catalan numbers are strictly monotonically increasing for $N \geq 1$, so that $\frac{1}{2}C_{(N-1)/2} + \frac{1}{2} \leq \frac{1}{2}C_{N-1} = \frac{1}{2}C_{N-1} C_0$ for $N \geq 2$.

We then have $\frac{1}{4} b_{N} b_1 + \frac{1}{4}b_{(N+1)/2} + \frac{1}{2}\leq C_{N-1} C_0$ and $\frac{1}{2}b_{N+1} \leq \sum_{i=1}^N C_{N-i}C_{i-1} = C_N$, and the induction is complete.

To complete the proof that $b_s\leq C_s$ for $s \geq 1$, we proceed again by induction. We note that the result holds in the base case $s=1$ ($1 = b_1 \leq C_1 = 1$) and $s=2$ ($2 = b_2 \leq C_2 = 2$), and suppose that it holds for all $s \leq N$, $N \geq 2$. Then 
\begin{eqnarray*} 
b_{N+1} &= & \bigg( \sum_{i=2}^{N} \frac{1}{2} b_{N+1-i}b_i \bigg) + \frac{1}{2} b_N b_1 + \frac{1}{2} b_{(N+1)/2} + 1\\
&\leq & \bigg( \sum_{i=2}^{N} C_{N+1-i}C_{i-1} \bigg) + \frac{1}{2} b_N b_1 + \frac{1}{2} b_{(N+1)/2} + 1. 
\end{eqnarray*}
We have $\frac{1}{2} b_N b_1 \leq \frac{1}{2} C_N C_0$ by the inductive hypothesis, and $\frac{1}{2}b_{(N+1)/2}+1 \leq  C_{(N-1)/2} + 1 \leq C_N =  C_N C_0$ by the earlier $\frac{1}{2}b_{s+1} \leq C_s$ and the strict monotonicity of the $C_N$ for $N \geq 2$. Then $\frac{1}{2} b_N b_1 + \frac{1}{2} b_{(N+1)/2} + 1 \leq\frac{3}{2} C_N C_0$ and $b_{N+1} \leq \big(\sum_{i=2}^N C_{N+1-i} C_{i-1}\big)+\frac{3}{2}C_NC_0 = \big(\sum_{i=1}^{N} C_{N+1-i} C_{i-1} \big)+\frac{1}{2}C_NC_0 < \big(\sum_{i=1}^{N+1} C_{N+1-i} C_{i-1}\big) = C_{N+1}$.
\end{proof}
\begin{corollary}
    The radius of convergence $\beta$ for $B(z)$ is positive, and in particular, $\frac{1}{4} \leq \beta \leq \rho$.
    \label{cor: beta geq 1/4}
\end{corollary}

\begin{proof}
We have seen that $\beta \leq \rho \approx 0.4027$ because the rooted binary perfect phylogenies include the rooted binary unalabeled trees, whose generating function has radius of convergence $\rho$. 

For the lower bound, because the Catalan generating function $C(z)=(1- \sqrt{1-4z})/(2z)$ has radius of convergence $\frac{1}{4}$, it follows from Lemma~\ref{lem:bncn} that $B(z) < C(z)$ for $0 < z < \frac{1}{4}$, so that the generating function $B(z)$ for the smaller sequence $\{b_n\}$ has radius of convergence $\beta \geq \frac{1}{4}$.
\end{proof}
\begin{theorem}
\label{thm: binary asymptotics}
The number $b_s$ of rooted binary perfect phylogenies with sample size $s$ has asymptotic growth 
\begin{equation} 
b_s \sim [\gamma/(2\sqrt{\pi})](1/r)^s s^{-3/2} \approx\frac{0.3519 (3.2599)^s}{s^{3/2}},
\label{eq:approx}
\end{equation}
where $\gamma \approx 1.2476$ and $r \approx 0.3068$ are constants. 
\end{theorem}
Because $B(z)$ is written in terms of $B(z^2)$ in Proposition~\ref{prop:Bs gen func}, the proof of Theorem~\ref{thm: binary asymptotics} relies on methods for generating functions defined implicitly. We use the smooth implicit-function schema in Theorem VII.3 from \cite[pp.~467-468]{flajolet_sedgewick_2009}. According to this theorem, we begin with an implicitly defined generating function $y(z) = \sum_{n \geq 0} y_n z^n$ that takes the form $y(z)=G\big(z,y(z) \big)$. Suppose that $y(z)$ is analytic at 0, $y_0 = 0$, and $y_n \geq 0$. Suppose also that
\begin{enumerate}[i.]
\item  $G(z,w) =\sum_{m,n\geq 0}g_{m,n} z^m w^n $ is analytic in a neighborhood of $(z,w)=(0,0)$.
\item $G(z,w)$ has coefficients $g_{m,n} \geq 0$ with $g_{0,0}=0$, $g_{0,1}\neq 1$, and $g_{m,n}>0$ for some $(m,n)$ with $n\geq2$. 
\item There exists some point $(z,w)=(r,s)$ in the analytic portion of the domain around $(0,0)$, such that $G(r,s)=s$ and $G_w(r,s)=1$.
\end{enumerate}
Then $[z^n]y(z)$ grows with $[\gamma/(2\sqrt \pi)](1/r)^n n ^{-3/2}$, where $\gamma=\sqrt{{2rG_z(r,s)}/{G_{ww}(r,s)}}$.

\begin{proof}
We verify that $B(z)$ belongs to the smooth implicit-function schema. Eq.~\eqref{eq:final_generating} gives the implicitly defined generating function. Write $G(z,w)= \frac{1}{2}w^2+\frac{1}{2}B(z^2) + \frac{z}{1-z}=\sum_{m,n \geq 0} g_{m,n}z^mw^n$. We prove $G(z,w)$ satisfies the required conditions.
\begin{enumerate}[i.]
\item We show $G(z,w)$ is analytic in a neighborhood of $(0,0)$. First note that $w^2/2$ is analytic for $|w| <\infty$. Next, for $\beta$ the radius of convergence of $B(z)$, $\frac{1}{2}B(z^2)$ is analytic for $|z|<\sqrt{\beta}$. Finally, $\frac{z}{1-z}$ is analytic for $z\neq 1$. Hence, noting $\beta < 1$, $G(z,w)$ is analytic for $|w|<\infty$ and $|z|<\sqrt{\beta}$.

\item  For the conditions on $g_{m,n}$, we examine the expansion of $G(z,w)$, and observe $g_{0,0} = 0$, $g_{0,1} = 0 \neq 1$, and $g_{0,2} = \frac{1}{2} > 0$. Each $g_{m,n}$ satisfies $g_{m,n} \geq 0$ for $m,n \geq 0$, as $B(z^2)$ and $z/(1-z)$ have nonnegative coefficients.

\item We show there exists a solution to the characteristic system 
$$G_w(z,w)=1, \, G(z,w)=w.$$ 
We first note that $G_w(z,w)=w$, so $G_w(r,s)=1$ is satisfied for $s=1$. Thus, we need only show that there exists $r$ such that
\begin{equation}
G(r,1)=\frac{1}{2}+\frac{1}{2}B(r^2)+\frac{r}{1-r} = 1.
\label{eq:Gr1}
\end{equation}
Consider $G(z,1)$ for $z$ real. It suffices to show that a real solution to eq.~\eqref{eq:Gr1} exists, a value of $r$ with $|r|\leq \sqrt{\beta}$ such that $G(r,1)=\frac{1}{2}$ in the domain where $G(z,w)$ is analytic. Restricting our attention to the positive, real line, we note that:
\begin{enumerate}[1.]
\item $G(z,1)$ is a monotonically increasing function for real $z>0$, as it is a sum of power series with nonnegative coefficients.
\item $G(0,1)=\frac{1}{2}$, as neither $B(z^2)$ nor $\frac{z}{1-z}$ has a constant term.
\item $G(\frac{1}{3},1)> 1$, as $\frac{1}{3}/(1-\frac{1}{3})=\frac{1}{2}$ and $B\big((\frac{1}{3})^2\big)>0$ (because $B(z^2)$ is strictly monotonically increasing for real $z>0$ and  $B(0)=0$). Note also that $\frac{1}{3} < \frac{1}{2} \leq \sqrt{\beta}$ by Corollary~\ref{cor: beta geq 1/4}.
\end{enumerate}
We conclude that there exists some $r$, $0 \leq r < \frac{1}{3} < \sqrt{\beta}$ such that $G(r,1)=1$.
\end{enumerate}

\vskip .2cm

We have therefore shown that $B(z)$ belongs to the smooth implicit function schema. The smooth implicit function schema tells us that the same $r$ that solves the characteristic system is indeed the radius of convergence of $B(z)$. With $s=1$, we solve eq.~\eqref{eq:Gr1} for $r$ numerically. We approximate $B(r^2)$ using the terms in Table~\ref{table:perfect}: $B(r^2) \approx \sum_{i=1}^{14} b_i (r^2)^i = 1r^2 + 2 r^4 + 3 r^6 + 7 r^8 + \ldots +109419 r^{28}$. Numerically solving for the positive, real root, we obtain $r\approx 0.306760104888$.

To compute the constant $\gamma$, we use $G_z(r,s)= rB'(r^2) + {1}/{(1-r)^2}$ and $G_{ww}(r,s)=1$. We approximate $B'(r^2)$ by the terms in Table~\ref{table:perfect}, 
\begin{equation*}
B'(r^2) \approx \sum_{i=1}^{14} ib_{i} (r^2)^{i-1} = (1\cdot 1) r^0 + (2\cdot2)r^2 + (3\cdot3)r^4 + (4\cdot7)r^6  + \ldots + (14\cdot109419)r^{26}.
\end{equation*}
We obtain $B'(r^2) \approx 1.4871$. Then
\begin{eqnarray*}
\gamma & = & \sqrt{2r \bigg[{r B'(r^2) + \frac{1}{(1-r)^2} \bigg]}} \approx 1.2476, \nonumber
\end{eqnarray*}
from which $b_s \sim [{\gamma}/{(2\sqrt{\pi})}](1/r)^s s^{-3/2} \approx {0.3519 (3.2599)^s}/{s^{3/2}}$.
\end{proof}

Figure~\ref{fig:3approx} plots the logarithm of the exact number of rooted binary perfect phylogenies $b_s$ from eq.~\eqref{eq:enumeration_alln} alongside the logarithm of the asymptotic growth from Theorem~\ref{thm: binary asymptotics}. We can observe, for example, that the asymptotic approximation $0.3519(3.2599)^{60}/60^{3/2}$ gives $4.6930\times 10^{27}$; the exact value is $4,753,678,474,171,125,902,623,929,051$.

\begin{figure}[tbp]
\centering
\includegraphics[width=7.4cm]{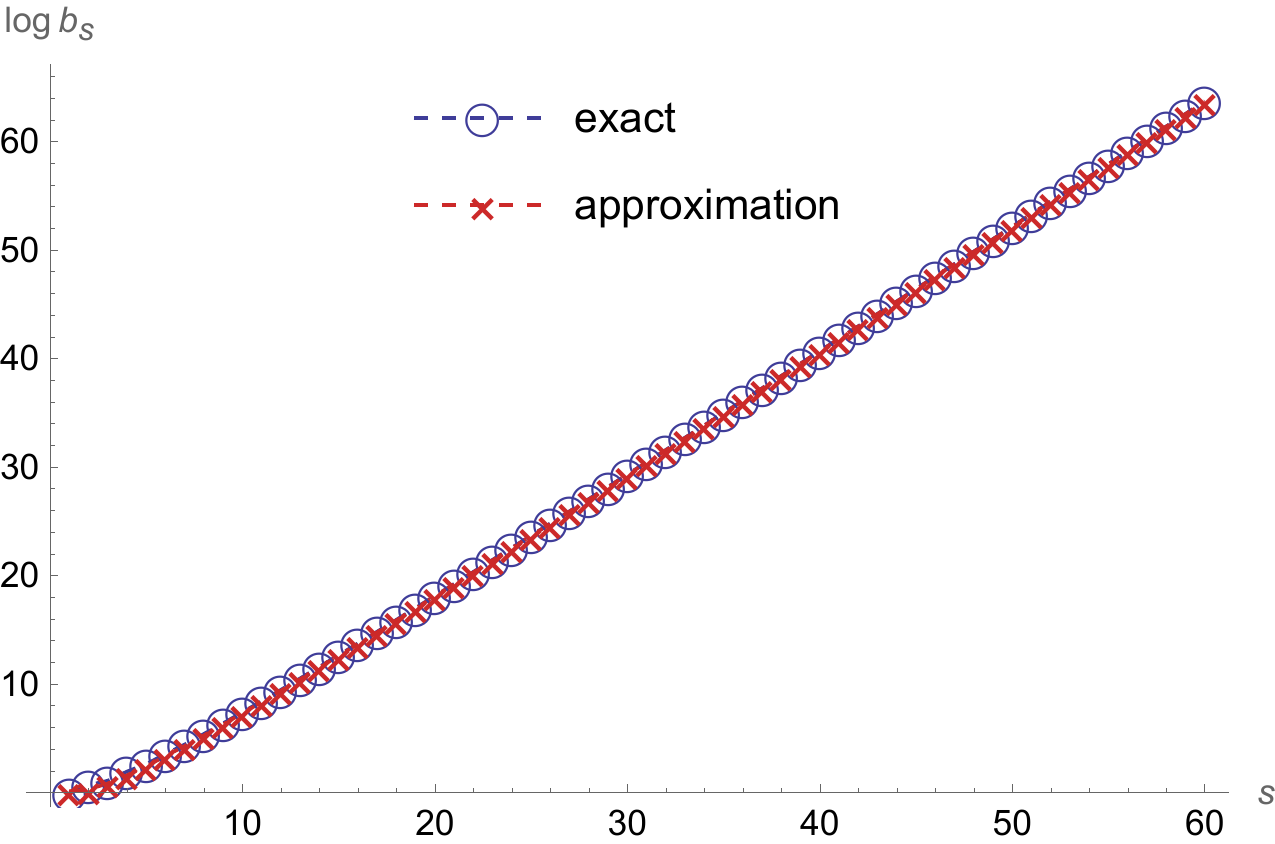}
\vspace{-.3cm}
\caption{The number $b_s$ of perfect phylogenies with sample size $s$. Exact values are computed from eq.~\eqref{eq:enumeration_alln}. The asymptotic approximation is computed from eq.~\eqref{eq:approx}.} 
\label{fig:3approx}
\end{figure}

\section{Rooted binary perfect phylogenies with sample size \texorpdfstring{$s$}{s} and \texorpdfstring{$n$}{n} leaves}\label{sec:s,n}

Having enumerated all rooted binary perfect phylogenies with sample size $s$, we now decompose the enumeration across perfect phylogenies with different numbers of leaves. A (non-empty) perfect phylogeny with sample size $s$ must possess a number of leaves in $[1,s]$. In the case that the number of leaves $n$ is equal to $s$, the rooted binary perfect phylogenies are simply rooted binary unlabeled trees, as each leaf has multiplicity 1. 

\subsection{Enumeration}

We generalize to consider all ordered pairs $(s,n)$ with $1 \leq n \leq s$. Let $b_{s,n}$ be the number of rooted binary perfect phylogenies with sample size $s$ and $n$ leaves, where $b_{s,n}=0$ if $s<n$, or $s\notin \mathbb{N}$, or $n\notin \mathbb{N}$.

\begin{proposition} The number $b_{s,n}$ of rooted binary perfect phylogenies with sample size $s$ and $n$ leaves, $1 \leq n \leq s$, satisfies \\
(i) $b_{s,1}=1$ for all $s \geq 1$. \\
(ii) For $(s,n)$ with $s \geq n \geq 2$,
\label{prop:b_s,n recursion}
\begin{equation}
b_{s,n}=
\begin{cases}
 \sum\limits_{j=1}^{n-1} \sum\limits_{i=j}^{s-n+j} \frac{1}{2} b_{s-i,n-j} \, b_{i,j}, & s \textnormal{ odd or } n\textnormal{ odd,} \\
\bigg( \sum\limits_{j=1}^{n-1} \sum\limits_{i=j}^{s-n+j} \frac{1}{2} b_{s-i,n-j} \, b_{i,j}\bigg)+\frac{1}{2}b_{s/2,n/2}, & s \textnormal{ even and } n \textnormal{ even}.\\
\end{cases} 
\label{eq: binary_rec}
\end{equation}
\end{proposition}
In eq.~\eqref{eq: binary_rec}, the index $i$ counts the sample size assigned to the right subtree and the index $j$ counts its number of leaves. The left subtree then has sample size $s-i$, with $n-j$ leaves. 

We observe that in the case $s=n$, in which a perfect phylogeny has all leaf multiplicities equal to 1, the recursion recovers eq.~\eqref{eq:wedderburn_rec}. In this case, we must have $i=j$. Because $s=n$, the cases become cases for $n$ odd or $n$ even. Because $n=s$ and $j=i$, we obtain for $n \geq 3$ odd:
\begin{equation*}
b_{n,n} = \sum\limits_{j=1}^{n-1} \frac{1}{2}b_{n-j,n-j} b_{j,j} = \sum\limits_{j=1}^{n-1} \frac{1}{2}u_{n-j} \, u_{j} = u_n.
\end{equation*}
The case for even $n \geq 2$ reduces to 
\begin{eqnarray*} 
b_{n,n} & = & \bigg( \sum\limits_{j=1}^{n-1} \frac{1}{2}b_{n-j,n-j} \, b_{j,j} \bigg) +\frac{1}{2} b_{n/2,n/2} = \bigg( \sum\limits_{j=1}^{n-1} \frac{1}{2}u_{n-j} \, u_{j} \bigg) +\frac{1}{2} u_{n/2} = u_n.
\end{eqnarray*}

\begin{proof}
We count rooted binary perfect phylogenies with sample size $s$ and $n$ leaves by considering all partitions of the sample and leaves into left and right subtrees. We index the sample size of the right subtree by $i$ and the number of leaves of the right subtree by $j$. 

The right subtree has sample size $i \geq j$. Because the left subtree has $n-j$ leaves, it has sample size at least $n-j$, so that the right subtree has sample size at most $i \leq s - (n-j)$.

Given the number of leaves $j$ for the right subtree, $1 \leq j \leq n-1$, and sample size $i$, $j \leq i \leq s-n+j$, if $j \neq \frac{n}{2}$ or $i \neq \frac{s}{2}$ or both, then we count $b_{s-i,n-j} \, b_{i,j}$ perfect phylogenies: $b_{i,j}$ for the right subtree and $b_{s-i,n-j}$ for the left. Each distinct perfect phylogeny is obtained twice, the second time with the left and right subtrees transposed.

If the sample size $s$ and number of leaves $n$ are both even, then we count both the $\binom{ b_{s/2,n/2}}{2} $ perfect phylogenies with two distinct subtrees of sample size $\frac{s}{2}$ and $\frac{n}{2}$ leaves as well as the $b_{s/2,n/2}$ perfect phylogenies with two identical subtrees, as in eq.~\eqref{eq:enumeration_alln} and eq.~\eqref{eq: binary_rec}.
\end{proof}

Figure~\ref{fig:4} shows an example of the enumeration, considering all possible rooted binary perfect phylogenies with $(s,n)=(8,6)$. Table~\ref{table:binary_counts} gives the values of $b_{s,n}$ for small $(s,n)$, illustrating that the sums $\sum_{n=1}^s b_{s,n}$ agree with the values obtained for $b_s$ via Proposition~\ref{prop:bs recursion}.


For fixed small $n$, $b_{s,n}$ can be stated in closed form. First, $b_{s,1}=1$ for $s \geq 1$. We obtain a sequence of corollaries of Proposition~\ref{prop:b_s,n recursion}.
\begin{corollary} 
For $s \geq 2$, the number $b_{s,2}$ of rooted binary perfect phylogenies with $n=2$ leaves is
\label{cor:s2}
\begin{equation}
b_{s,2}=\bigfloor{ \frac{s}{2} }.
\label{eq:s2}
\end{equation}
\end{corollary}

\begin{table}[tb]
\centering
\begin{tabular}{c c c c c c c c c c c c c c} 
\toprule
&\multicolumn{11}{c}{Number of leaves (n)}   \\        
 \cmidrule(r){2-12}
Sample size ($s$) & 11 & 10 & 9 & 8 & 7 & 6 & 5 & 4 & 3 & 2 & 1 & Total\\ 
 \midrule
1  & &    &     &     &     &     &     &    &    &    & 1 &    1 \\ [.12ex]
2  & &    &     &     &     &     &     &    &    &  1 & 1 &    2 \\ [.12ex]
3  & &    &     &     &     &     &     &    &  1 &  1 & 1 &    3 \\ [.12ex]
4  & &    &     &     &     &     &     &  2 &  2 &  2 & 1 &    7 \\ [.12ex]
5  & &    &     &     &     &     &   3 &  4 &  4 &  2 & 1 &   14 \\ [.12ex] 
6  & &    &     &     &     &   6 &   9 & 10 &  6 &  3 & 1 &   35 \\ [.12ex] 
7  & &    &     &     &  11 &  20 &  24 & 17 &  9 &  3 & 1 &   85 \\ [.12ex] 
8  & &    &     &  23 &  46 &  61 &  49 & 30 & 12 &  4 & 1 &  226 \\ [.12ex] 
9  & &    &  46 & 106 & 152 & 138 &  93 & 44 & 16 &  4 & 1 &  600 \\ [.12ex] 
10 & & 98 & 248 & 386 & 387 & 290 & 157 & 66 & 20 &  5 & 1 & 1658 \\ [.12ex]
11 & 207 & 582 & 974 & 1072 & 878 & 535 & 253 & 90 & 25 & 5 & 1 & 4622 \\ [.12ex]
12 & 1376 & 2473 & 2951 & 2633 & 1774 & 939 & 383 & 124 & 30 & 6 & 1 & 13141 \\ [.12ex]
13 & 6262 & 8061 & 7763 & 5727 & 3340 & 1534 & 562 & 160 & 36 & 6 & 1 & 37699 \\ [.12ex]
14 & 21899 & 22657 & 18119 & 11551 & 5881 & 2420 & 792 & 208 & 42 & 7 & 1 & 109419 \\ [.12ex]
\bottomrule
\end{tabular}
\caption{\label{table:binary_counts}The number $b_{s,n}$ of rooted binary perfect phylogenies with sample size $s$ and $n$ leaves. Entries are obtained using eq.~\eqref{eq: binary_rec}; the ``total'' is $b_s = \sum_{n=1}^s b_{s,n}$. The total follows A113822; $b_{s,3}$ follows A002620. The main diagonal and its subdiagonal follow A001190 and A085748. For completeness, $b_{12,12}=451$, $b_{13,12}=3264$, $b_{13,13}=983$, $b_{14,12}=15886$, $b_{14,13}=7777$, and $b_{14,14}=2179$.}
\end{table}
\begin{proof}
From Proposition~\ref{prop:b_s,n recursion}, we have 
\begin{equation*}
b_{s,2}=
\begin{cases}
\sum\limits_{i=1}^{s-1}\frac{1}{2}b_{s-i,1}b_{i,1}, & \text{odd } s \geq 3,\\
\bigg(\sum\limits_{i=1}^{s-1}\frac{1}{2}b_{s-i,1}b_{i,1}\bigg)+\frac{1}{2}b_{s/2,1}, & \text{even } s \geq 2.
\end{cases}
\end{equation*}
Because $b_{s,1}=1$ for all $s \geq 1$, we obtain $b_{s,2}=\frac{s-1}{2}$ for odd $s \geq 3$, and $b_{s,2}=\frac{s}{2}$ for even  $s \geq 2$. Summarizing the odd and even cases in one expression, $b_{s,2}=\floor{ \frac{s}{2}}$.
\end{proof}

\begin{corollary}
For $s \geq 3$, the number $b_{s,3}$ of rooted binary perfect phylogenies with $n=3$ leaves is
\label{cor:s3}
\begin{equation}
b_{s,3}= \bigfloor{ \frac{s-1}{2} } \bigg \lceil \frac{s-1}{2} \bigg \rceil.
\label{eq:s3}
\end{equation}
\end{corollary}
\begin{proof} Using Proposition~\ref{prop:b_s,n recursion},
\begin{eqnarray*}
b_{s,3} & = & \sum\limits_{j=1}^2\sum\limits_{i=j}^{s-3+j}\frac{1}{2}b_{s-i,3-j}b_{i,j} = \sum\limits_{i=2}^{s-1} b_{s-i,1}b_{i,2},
\end{eqnarray*}
noting $\sum_{i=1}^{s-2}\frac{1}{2}b_{s-i,2}b_{i,1}=\sum_{i=2}^{s-1}\frac{1}{2}b_{s-i,1}b_{i,2}$ by an exchange of $s-i$ and $i$.

Using Corollary~\ref{cor:s2}, $b_{s,3}=\sum_{i=2}^{s-1} \floor{ \frac{i}{2} }$. The summation yields 
\begin{eqnarray*}
b_{s,3} = 
\begin{cases}
2\big(1+2+\dots +\frac{s-3}{2}\big)+\frac{s-1}{2} = \big( \frac{s-1}{2} \big) \big(\frac{s-1}{2} \big), & \text{odd } s \geq 3, \\
2\big(1+2+\dots +\frac{s-2}{2} \big) = \big( \frac{s-2}{2} \big) \big( \frac{s}{2} \big), & \text{even } s \geq 4.
\end{cases}
\end{eqnarray*}
We can summarize both cases in the single expression $b_{s,3} = \floor{ \frac{s-1}{2}} \lceil \frac{s-1}{2} \rceil$.
\end{proof}
For small $s$, the rooted binary perfect phylogenies with $n=3$ leaves appear in Figure~\ref{fig:5caterpillar}.

\begin{corollary} For $s\geq 4$, the number $b_{s,4}$ of rooted binary perfect phylogenies with $n=4$ leaves is 
\begin{equation}
b_{s,4}= 
\begin{cases}
\frac{(s-1)(s-3)(5s-1)}{48}, & \text{odd } s \geq 5, \\
\frac{s(s-2)(5s-11)}{48} + \frac{1}{2} \floor{ \frac{s}{4}},
& \text{even } s \geq 4.
\end{cases}
\label{eq:s4}
\end{equation}
\label{cor: s4}
\end{corollary}
\begin{proof}
Using eq.~\eqref{eq: binary_rec}, we see that: 
\begin{equation*}
b_{s,4}=
\begin{cases}
\label{eq:s4_rec}
\sum\limits_{j=1}^{3} \sum\limits_{i=j}^{s-4+j} \frac{1}{2} b_{s-i,4-j} \, b_{i,j}, & \textnormal{odd } s \geq 5,\\
\bigg( \sum\limits_{j=1}^{3} \sum\limits_{i=j}^{s-4+j} \frac{1}{2} b_{s-i,4-j} \, b_{i,j}\bigg)+\frac{1}{2}b_{s/2,2}, & \textnormal{even } s \geq 4.
\end{cases} 
\end{equation*}
The outer sum considers values of 1, 2, and 3 for the sample size $j$ assigned to the right subtree. Assigning $j=3$ gives the same inner sum as $j=1$, as $\sum_{i=3}^{s-1} \frac{1}{2} b_{s-i,1} b_{s,3} = \sum_{i=1}^{s-3} \frac{1}{2} b_{s-i,3} b_{s,1}$ by an exchange of $i$ and $s-i$. Noting that $b_{s,1}=1$ for $s \geq 1$, the problem becomes:
\begin{equation*}
b_{s,4} = \begin{cases}
\label{eq:s4_rec2}
\sum\limits_{i=1}^{s-3} b_{s-i,3} + \sum\limits_{i=2}^{s-2} \frac{1}{2}b_{s-i,2} b_{i,2}, & \textnormal{odd } s \geq 5,\\
\bigg( \sum\limits_{i=1}^{s-3} b_{s-i,3} + \sum\limits_{i=2}^{s-2} \frac{1}{2}b_{s-i,2} b_{i,2} \bigg) +\frac{1}{2}b_{s/2,2}, & \textnormal{even } s \geq 4.
\end{cases} 
\end{equation*}

\begin{figure}[tb]
\centering
\includegraphics[width=\textwidth]{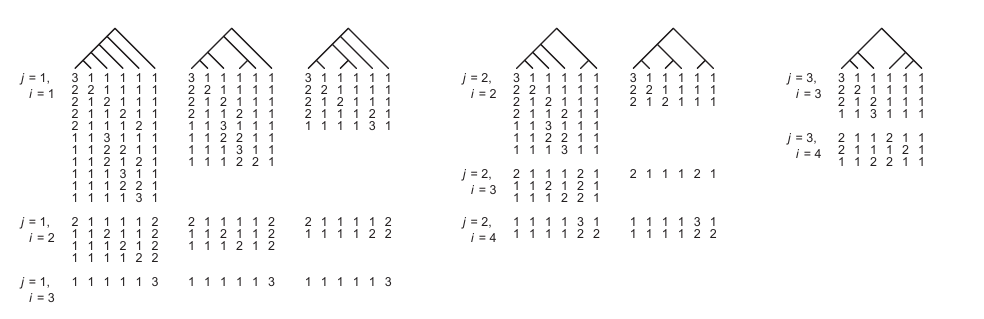}
\vspace{-.5cm}
\caption{The enumeration of all $b_{8,6}=61$ rooted binary perfect phylogenies with sample size $s=8$ and $n=6$ leaves. The number of leaves in the right subtree is indicated by $j$, and $i$ indicates the sample size for the right subtree.}
\label{fig:4}
\end{figure}

For $s$ odd, $s-3$ is even and $s-2$ is odd. We use Corollaries~\ref{cor:s2} and \ref{cor:s3} to obtain the summands, resolving floor and ceiling functions separately for odd and even quantities. The sums are completed using $\sum_{k=1}^n k = n(n+1)/2$ and $\sum_{k=1}^n k^2 = n(n+1)(2n+1)/6$. Then
\begin{eqnarray}
\sum\limits_{i=1}^{s-3} b_{s-i,3} 
 &=&\sum\limits_{i=1}^{s-3} \bigfloor{ \frac{s-i-1}{2} } \bigg\lceil \frac{s-i-1}{2} \bigg\rceil \nonumber\\
& = & \sum\limits_{k=1}^{\frac{s-3}{2}}
\bigfloor{ \frac{s-(2k-1)-1}{2} } \bigg\lceil \frac{s-(2k-1)-1}{2} \bigg\rceil + \bigfloor{ \frac{s-2k-1}{2} } \bigg\lceil \frac{s-2k-1}{2} \bigg\rceil
\nonumber\\
& = & \sum\limits_{k=1}^{\frac{s-3}{2}}
\bigg( \frac{s-2k-1}{2} \bigg) \bigg( \frac{s-2k+1}{2} \bigg) + \bigg( \frac{s-2k -1}{2} \bigg)^2 \nonumber \\
\label{eq:odd3}
& = & {(s-1)(s-3)(2s-1)}/{24}. \\
\sum\limits_{i=2}^{s-2} b_{s-i,2} \, b_{i,2}
 &=&\sum\limits_{i=2}^{s-2} \bigfloor{ \frac{s-i}{2} } \bigfloor{ \frac{i}{2} } 
 \nonumber\\
& = & \sum_{k=1}^{\frac{s-3}{2}} \bigfloor{ \frac{s-(2k+1)}{2} } \bigfloor{ \frac{2k+1}{2} } 
+ \bigfloor{ \frac{s-(2k)}{2} } \bigfloor{ \frac{2k}{2} } \nonumber \\
& = & 
\sum_{k=1}^{\frac{s-3}{2}} \bigg(\frac{s-2k-1}{2}\bigg) k 
+ \bigg(\frac{s-2k-1}{2}\bigg) k \nonumber \\
\label{eq:odd2}
& = & {(s+1)(s-1)(s-3)}/{24}. 
\end{eqnarray}
Summing eq.~\eqref{eq:odd3} and half of eq.~\eqref{eq:odd2}, we obtain the result in eq.~\eqref{eq:s4} for odd $s$.

The even case is similar, except that now $s-3$ is odd and $s-2$ is even.
\begin{eqnarray}
\sum\limits_{i=1}^{s-3} b_{s-i,3} 
& = &\sum\limits_{i=1}^{s-3} \bigfloor{ \frac{s-i-1}{2} } \bigg\lceil \frac{s-i-1}{2} \bigg\rceil \nonumber\\
& = & \bigfloor{ \frac{s-(s-3)-1}{2} } \bigg\lceil \frac{s-(s-3)-1}{2} \bigg\rceil \nonumber \\
& + &  \sum\limits_{k=1}^{\frac{s-4}{2}}
\bigfloor{ \frac{s-(2k-1)-1}{2} } \bigg\lceil \frac{s-(2k-1)-1}{2} \bigg\rceil + \bigfloor{ \frac{s-2k-1}{2} } \bigg\lceil \frac{s-2k-1}{2} \bigg\rceil
\nonumber\\
& = & 1 + \sum\limits_{k=1}^{\frac{s-4}{2}}
\bigg( \frac{s-2k}{2} \bigg)^2 + \bigg( \frac{s-2k -2}{2} \bigg) \bigg( \frac{s-2k}{2} \bigg) \nonumber \\
\label{eq:even3} 
& = & {s(s-2)(2s-5)}/{24}.
\end{eqnarray}
\begin{eqnarray}
\sum\limits_{i=2}^{s-2} b_{s-i,2} \, b_{i,2} & = &\sum\limits_{i=2}^{s-2} \bigfloor{ \frac{s-i}{2} } \bigfloor{ \frac{i}{2} } \nonumber\\
& = & \bigfloor{ \frac{s-(s-2)}{2}  } \bigfloor{ \frac{s-2}{2} } + \sum_{k=1}^{\frac{s-4}{2}} \bigfloor{ \frac{s-(2k+1)}{2} } \bigfloor{ \frac{2k+1}{2} } 
+ \bigfloor{ \frac{s-(2k)}{2} } \bigfloor{ \frac{2k}{2} } \nonumber \\
& = & \frac{s-2}{2} + \sum_{k=1}^{\frac{s-4}{2}} \bigg(\frac{s-2k-2}{2}\bigg) k + \bigg(\frac{s-2k}{2}\bigg) k \nonumber \\
\label{eq:even2}
& = & {s(s-1)(s-2)}/{24}. 
\end{eqnarray}
We obtain eq.~\eqref{eq:s4} for even $s$ by summing eq.~\eqref{eq:even3}, half of eq.~\eqref{eq:even2}, and $\frac{1}{2}b_{s/2,s} = \frac{1}{2} \floor{ \frac{s}{4} }$.
\end{proof}

\begin{figure}[tbp]
\centering
\includegraphics[width=\linewidth]{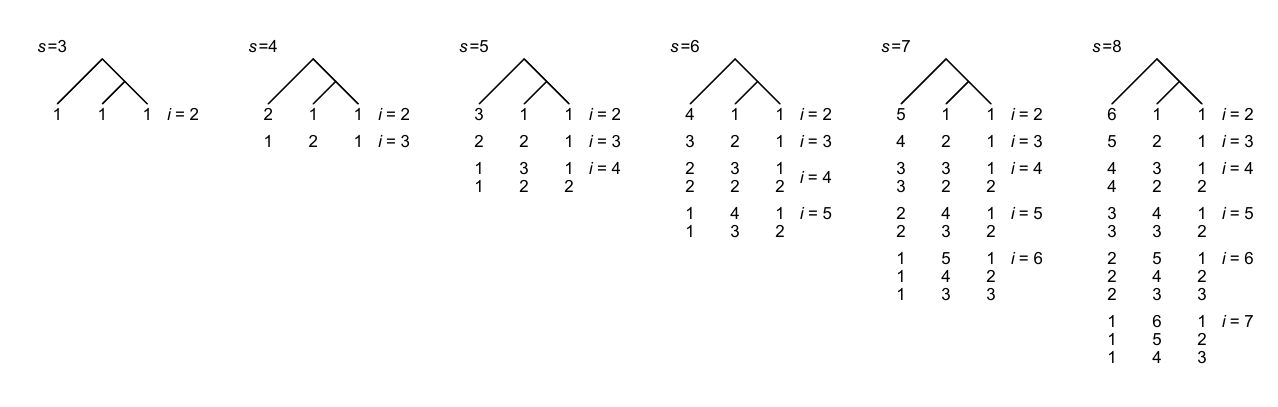}
\vspace{-1cm}
\caption{The $b_{s,3}$ rooted binary perfect phylogenies with $n=3$ leaves, for each $s$ from 3 to 8, as obtained by Proposition~\ref{prop:b_s,n recursion}. The value of $i$ indicates the sample size for the right subtree.} 
\label{fig:5caterpillar}
\end{figure}

\subsection{Generating function}
\label{sec:gf_leaves}

Beyond the closed-form expressions for the cases of $n=2$, 3, and 4, for a specific value of $n$ more generally, we can obtain a generating function for the sequence $b_{s,n}$ describing the number of rooted binary perfect phylogenies with fixed $n$ and increasing values of the sample size $s$.

Denote by $B_n(z)$ the generating function describing the sequence of values for the number of rooted binary perfect phylogenies with $n$ leaves and increasing sample size $s$: 
$$B_n(z)=\sum_{s=1}^\infty b_{s,n}z^s.$$
Recall that $b_{s,n} > 0$ only for integers $(s,n)$ with $s \geq n \geq 1$. We define $B_n(z)=0$ for non-integer values of $n$. 
\begin{proposition}
\label{lemma: bn gf}
The generating function $B_n(z)$ for the number $b_{s,n}$ of rooted binary perfect phylogenies with sample size $s$ and $n$ leaves satisfies \\
(i) $B_1(z)=\frac{z}{1-z}$. \\
(ii) For $n\geq 2$,
\begin{equation}
B_n(z)=\bigg( \sum_{j=1}^{n-1} \frac{1}{2} B_{n-j}(z) \, B_j(z)\bigg)+\frac{1}{2}B_{n/2}(z^2).
\label{eq:n leaves gf}
\end{equation}
\end{proposition}
\begin{proof}
(i) For $n=1$, for all $s\geq 1$, $b_{s,1}=1$. Hence, $B_1(z)=\frac{z}{1-z}$. \\
(ii) For $n\geq 2$, we use the recursive eq.~\eqref{eq: binary_rec}, 
\begin{align*}
B_n(z) &=\sum_{s=n}^\infty \bigg(\sum_{j=1}^{n-1}\sum_{i=j}^{s-n+j}\frac{1}{2}b_{s-i,n-j}b_{i,j} \bigg)z^s + \sum_{s=n}^\infty \frac{1}{2}b_{s/2,n/2}z^s \\
&=\sum_{s=n}^\infty\bigg(\sum_{j=1}^{n-1}\sum_{i=j}^{s-n+j}\frac{1}{2}b_{s-i,n-j}z^{s-i}b_{i,j} z^i \bigg)+ \sum_{s=n}^\infty \frac{1}{2}b_{s/2,n/2}z^s.
\end{align*}
We adjust the summation limits, noting that $b_{s,n}=0$ for $s < n$:
\begin{align*}
B_n(z) &= \sum_{s=2}^\infty \bigg(\sum_{j=1}^{n-1}\sum_{i=1}^{s-1}\frac{1}{2}b_{s-i,n-j}z^{s-i}b_{i,j} z^i \bigg)+\sum_{s=2}^\infty \frac{1}{2}b_{s/2,n/2}z^s\\
&=\bigg(\sum_{j=1}^{n-1} \frac{1}{2} \big(b_{1,n-j}b_{1,j} z^2+(b_{2,n-j}b_{1,j}+b_{1,n-j}b_{2,j}) z^3+\ldots \big)\bigg)+ \sum_{s=2}^\infty \frac{1}{2}b_{s/2,n/2}z^s.
\end{align*}
We can factor the first summation into a product of generating functions to get
\begin{align*}
B_n(z)&=\bigg(\sum_{j=1}^{n-1}\frac{1}{2}(b_{1,n-j}z^1+b_{2,n-j}z^2+\ldots)(b_{1,j} z^1+b_{2,j}z^2 +\ldots)\bigg)+\frac{1}{2}(b_{1,n/2}z^2+b_{2,n/2}z^4+\ldots)\\
&=\bigg( \sum_{j=1}^{n-1}\frac{1}{2}  B_{n-j}(z) \, B_j(z)\bigg)+\frac{1}{2}B_{n/2}(z^2),
\end{align*}
completing the proof.
\end{proof}
Iterating from the base case $B_1(z)=\frac{z}{1-z}$, we can write an explicit form for $B_n(z)$ for a fixed $n$. Using Proposition~\ref{lemma: bn gf}, we have 
\begin{align}
B_2(z)&=\frac{1}{2}B_1(z)^2+\frac{1}{2}B_1(z^2) \nonumber\\ 
&= \frac{1}{2}\frac{z^2}{(1-z)^2}+\frac{1}{2}\frac{z^2}{1-z^2}\nonumber\\
&=\frac{z^2}{(1-z)^2(1+z)}. \label{eq:s2 proof}
\end{align}
Next, using $B_2(z)$ from eq.~\eqref{eq:s2 proof}, we have
\begin{align} 
B_3(z)&=\frac{1}{2}B_1(z)\,B_2(z)+\frac{1}{2}B_2(z)\, B_1(z)\nonumber\\
&=\frac{z^3}{(1-z)^3(1+z)}.\label{eq:B3}
\end{align}
For $B_4(z)$, we use eq.~\eqref{eq:B3}
\begin{align}
B_4(z)&=\frac{1}{2} B_1(z) \, B_3(z)+ \frac{1}{2}B_2^2(z) + \frac{1}{2} B_3(z) \, B_1(z)+\frac{1}{2}B_2(z^2)\nonumber\\
&=\frac{z^4}{(1-z)^4(1+z)}+\frac{1}{2}\frac{z^4}{(1-z)^4(1+z)^2}+\frac{1}{2}\frac{z^4}{(1-z^2)^2(1+z^2)} \nonumber \\
&= \frac{z^4 (2 + 2z^2 + z^3)}{(1-z)^4(1+z)^2(1+z^2)}.
\label{eq:B4}
\end{align}
We can continue iteratively to get generating functions for larger $n$. 

\section{Rooted binary perfect phylogenies with a specific unlabeled tree shape}
\label{sec:treeshape}

We have counted perfect phylogenies with sample size $s$, with sample size $s$ and number of leaves $n$, and with number of leaves $n$. Each perfect phylogeny has an associated unlabeled tree shape; we now count the perfect phylogenies with sample size $s$ and a specific shape (with $n$ leaves). We begin with the case that the tree is a caterpillar. We then proceed to the general case of an arbitrary unlabeled tree shape.

\subsection{Caterpillar tree shape}
\label{sec:cat_shape}
A caterpillar tree with $n \geq 2$ leaves has exactly 1 cherry node. In other words, for $n \geq 3$, a caterpillar tree is constructed by adjoining a caterpillar tree with $n-1$ leaves and a single-leaf tree to a shared root.

Denote by $g_{s,n}$ the number of caterpillar rooted binary phylogenies with sample size $s$ and $n \geq 2$ leaves. We set $g_{s,n}=0$ if $s < n$ (or $s \notin \mathbb{N}$, or $n \notin \mathbb{N}$). 
\begin{proposition}
\label{prop:cat rec}
The number $g_{s,n}$ of rooted binary perfect phylogenies with caterpillar shape, sample size $s$, and $n$ leaves, $2 \leq n \leq s$, satisfies \\
(i) $g_{s,2}=\floor{\frac{s}{2}} $ for all $s \geq 2$. \\
(ii) For $(s,n)$ with $s \geq n\geq 3$,
\begin{equation}
g_{s,n}=\sum_{i=n-1}^{s-1} g_{i,n-1}=\sum_{i_1=n-1}^{s-1}\sum_{i_2=n-2}^{i_1-1} \ldots \sum_{i_{n-3}=3}^{i_{n-4}-1}\sum_{i_{n-2}=2}^{i_{n-3}-1} \bigfloor{ \frac{i_{n-2}}{2}}. 
  \label{eq:caterpillarcounts}
\end{equation}
\end{proposition}
\begin{proof} (i) Recognizing that the only tree shape with $n=2$ leaves is the 2-leaf caterpillar tree, we see that we already proved this result in Corollary~\ref{cor:s2}.

(ii) For $n\geq 3$, the left subtree of a caterpillar of size $n$ is a caterpillar of size $n-1$. We assign sample size $i$ to the left subtree, $n-1 \leq i \leq s-1$, and sample size $s-i$ to the single leaf in the right subtree, obtaining 
\begin{equation}
\label{eq:gsn}
    g_{s,n}=\sum_{i=n-1}^{s-1} g_{i,n-1} b_{s-i,1} = \sum_{i=n-1}^{s-1} g_{i,n-1}.
\end{equation}
Proceeding iteratively, we have
$$g_{s,n} = \sum_{i_1=n-1}^{s-1} g_{i_1,n-1} = \sum_{i_1=n-1}^{s-1} \sum_{i_2=n-2}^{i_1-1} g_{i_2,n-2} = \ldots = \sum_{i_1=n-1}^{s-1}\sum_{i_2=n-2}^{i_1-1} \ldots \sum_{i_{n-3}=3}^{i_{n-4}-1}\sum_{i_{n-2}=2}^{i_{n-3}-1} g_{i_{n-2},2}.$$
We apply the base case of $n=2$ to complete the proof.
\end{proof}

We can apply Proposition~\ref{prop:cat rec} with specific small values of $n$, completing the sum in eq.~\eqref{eq:caterpillarcounts}. The case of $n=3$ was obtained in  Corollary~\ref{cor:s3}, and we will write its solution in a different form. We proceed via calculations similar to those performed in obtaining Corollaries~\ref{cor:s2}-\ref{cor: s4}.

We use an approach that avoids summations that include floor and ceiling functions, as appeared in  the proofs of Corollaries~\ref{cor:s3} and \ref{cor: s4}. Separating the $n=2$ result $g_{s,2} = \lfloor \frac{s}{2} \rfloor$ (Corollary~\ref{cor:s2}) into cases for odd and even $s$, we can increase $n$ incrementally, observing from eq.~\eqref{eq:gsn} that for fixed $n$, $g_{s,n}$ as a function of $s$ can be written with odd and even cases, each consisting of a polynomial of degree $n-1$ in $s$. It is convenient to instead define the cases in terms of odd and even $s-n$. In particular, for $s-n$ even, we define $f_{s,n}^e$ to be the polynomial describing the number of caterpillar perfect phylogenies with sample size $s$ and $n$ leaves. Similarly, for $s-n$ odd, we define $f_{s,n}^o$ as the corresponding polynomial for the number of caterpillar perfect phylogenies with sample size $s$ and $n$ leaves. Both polynomials can be calculated for all $(s,n)$ with $s \geq n \geq 2$, but each represents the number of caterpillar perfect phylogenies only in its associated case. 

With these definitions, the number of caterpillar perfect phylogenies $g_{s,n}$ can be written in a form that is convenient for computation, containing only a single floor function.
\begin{proposition}
\label{prop:gsn}
For $(s,n)$ with $n\geq 2$ and $s\geq n$, the number of caterpillar perfect phylogenies with $n$ leaves and sample size $s$ is
\begin{equation}
    g_{s,n}=\floor{ f_{s,n}^e}= \bigfloor{ \bigg (\sum_ {i = n-1}^{s - 1} f_{i,n-1}^e \bigg) - \frac {s - n}{2^{n-1}} }.
    \label{eq:gsn recursion}
    \end{equation}
\end{proposition}
The proposition relies on a lemma.
\begin{lemma} 
\label{lemma:evenodd}
For $(s,n)$ with $n\geq 2$ and $s\geq n$, 
\begin{equation}
f_{s,n}^e-f_{s,n}^o=\frac{1}{2^{n-1}}.
\end{equation}
\end{lemma}
\begin{proof}
We proceed by induction on $n$. For the base case, $n=2$, by Corollary~\ref{cor:s2}, we have $g_{s,2}=\floor{\frac{s}{2}}$. Thus for $s$ even, $g_{s,2}=\frac{s}{2}$, and for $s$ odd, $g_{s,2}=\frac{s}{2}-\frac{1}{2}$. In other words, we have $f_{s,2}^e=\frac{s}{2}$ and $f_{s,2}^o=\frac{s-1}{2}$. It follows that $f_{s,2}^e-f_{s,2}^o=\frac{1}{2}$.

For the inductive step, suppose for $n\geq 3$ that $f_{s,n-1}^e-f_{s,n-1}^o=1/2^{n-2}$. If $s-(n-1)$ is even, then $f_{s,n-1}^e$ is an integer, with $g_{s,n-1}=f_{s,n-1}^e=\floor{f_{s,n-1}^e}$. If instead $s-(n-1) $ is odd, then $f_{s,n-1}^o = g_{s,n-1}$ is an integer, and by the inductive hypothesis, $\floor{f_{s,n-1}^e}=f_{s,n-1}^e-1/2^{n-2}=f_{s,n-1}^o$.

By eq.~\eqref{eq:gsn} and the inductive hypothesis, using $\mathbbm{1}_{\{x\}}=1$ if statement $x$ holds and $\mathbbm{1}_{\{x\}}=0$ otherwise, we obtain 
\begin{align*}
g_{s,n}=&\sum_{i=n-1}^{s-1} g_{i,n-1} =\sum_{i=n-1}^{s-1} \floor{f_{i,n-1}^e} \\
=& \bigg( \sum_{i=n-1}^{s-1} f_{i,n-1}^e \bigg) - \frac{1}{2^{n-2}} \sum_{i=n-1}^{s-1} \mathbbm{1}_{\{i-(n-1) \text{ is odd}\}}.
\end{align*}
We then use that 
$$
\sum_{i=n-1}^{s-1} \mathbbm{1}_{\{i-(n-1) \text{ is odd}\}} = \sum_{i=0}^{s-n} \mathbbm{1}_{\{i \text{ is odd}\}}=
    \begin{cases}
    \frac{s-n}{2}, & s-n \text { even,}\\
    \frac{s-(n-1)}{2}, & s-n \text { odd.}
    \end{cases}
$$
We then obtain expressions for $g_{s,n}$ in the case of even $s-n$ and odd $s-n$. Because $g_{s,n}=f_{s,n}^e$ for even $s-n$ and $g_{s,n}=f_{s,n}^o$ for odd $s-n$, we have 
\begin{align}
f_ {s,n}^e &= \bigg(\sum_ {i = n-1}^{s - 1} f_{i,n-1}^e \bigg) - \frac {s \
-n}{2^{n-1}},\label{eq:fe}\\
f_ {s,n}^o &=\bigg (\sum_ {i = n-1}^{s - 1} f_{i,n-1}^e \bigg) - \frac {s - \
(n-1)}{2^{n-1}}\label{eq:fo}.
\end{align}
Now we see that $f_{s,n}^e-f_{s,n}^o=1/2^{n-1}$, completing the induction. 
\end{proof}
\noindent \emph{Proof of Proposition~\ref{prop:gsn}.}
From Lemma~\ref{lemma:evenodd}, for each $n\geq 2$, $f_{s,n}^e$ exceeds $f_{s,n}^o$ by a quantity that is less than 1. Hence, for odd $s-n$, for which $g_{s,n}=f_{s,n}^o$ and $f_{s,n}^o$ is an integer, $\floor{f_{s,n}^e} = f_{s,n}^o = g_{s,n}$. For even $s-n$, $g_{s,n}=f_{s,n}^e$ and $f_{s,n}^e$ is an integer, so that $\floor{f_{s,n}^e} = f_{s,n}^e = g_{s,n}$. We conclude in both odd and even cases that $g_{s,n}= \floor{f_{s,n}^e}$, with $f_{s,n}^e$ specified by eq.~\eqref{eq:fe}.
$\square$

\vskip .2cm
We can then compute $g_{s,n}$ for the smallest $n$ by iteratively summing polynomials to calculate $f_{s,n}^e$ in eq.~\eqref{eq:fe}, taking the floor of the output. We present the first several functions $g_{s,n}$.
\begin{align}
g_{s,2}&=&\bigfloor{ \frac{s}{2}}\\
g_{s,3}&=&\bigfloor{ \frac{(s-1)^2}{4}}\\
g_{s,4}&=&\bigfloor{ \frac{s(s-2)(2s-5)}{24}} \label{eq: C4}\\
g_{s,5}&=&\bigfloor{ \frac{(s-1) (s-3) (s^2 - 4 s + 1)}{48}}\\
g_{s,6}&=&\bigfloor{ \frac{s(s-2)(s-4)(2s^2-13s+16)}{480}}\\
g_{s,7}&=&\bigfloor{ \frac{(s-1) (s-3 )^2 (s-5) (2s^2-12s+1 )}{2880}}\\
g_{s,8}&=&\bigfloor{ \frac{s(s-2)(s-4)(s-6)(4s^3-50s^2+176s-151)}{40320}}\\
g_{s,9}&=&\bigfloor{ \frac{(s-1)(s-3) (s-5) (s-7)  (s^4-16s^3+78s^2 - 112 s+3)}{80640}}.
\end{align}
The values of $g_{s,n}$ for $2 \leq n \leq s \leq 13$ appear in Table~\ref{table:caterpillars}. 

We next obtain a generating function for the number of perfect phylogenies with the fixed caterpillar shape with $n$ leaves, as $s$ increases.
\begin{proposition}
\label{prop:cat_gf}
    The generating function $G_n(z)$ for the number $g_{s,n}$ of rooted binary perfect phylogenies with sample size $s \geq n$ and the caterpillar topology with $n \geq 2$ leaves satisfies  
    \begin{equation}
        G_n(z)=\frac{z^n}{(1-z)^n (1+z)}.
        \label{eq:caterpillar_gf}
    \end{equation}
\end{proposition}
\begin{proof}
We proceed by induction. We obtained the result for $n=2$ in eq.~\eqref{eq:s2 proof}, as the caterpillar is the only shape with 2 leaves. Suppose the generating function for the number of rooted binary perfect phylogenies with sample size $s$ and the $n$-leaf caterpillar follows eq.~\eqref{eq:caterpillar_gf}. We apply eq.~\eqref{eq:gsn} to obtain the generating function associated with the caterpillar with $n+1$ leaves:
\begin{align*}
G_{n+1}(z)&=\sum_{s=1}^\infty g_{s,n+1} z^s\\
&=\sum_{s=1}^\infty \bigg(\sum_{i=n}^{s-1} g_{i,n} b_{s-i,1}\bigg) z^s \\
&=\sum_{s=1}^\infty \sum_{i=n}^{s-1} (g_{i,n} z^i)(b_{s-i,1} z^{s-i}).
\end{align*}

\begin{table}[tb]
\centering     
\begin{tabular}{cccccccccccccc}
\toprule
&\multicolumn{13}{c}{Number of leaves ($n$)} \\        
 \cmidrule(r){2-13}
Sample size ($s$) & 13 & 12 & 11 & 10 & 9 & 8 & 7 & 6 & 5 & 4 & 3 & 2 & Total \\
\midrule
 2  &   &    &     &     &     &     &     &     &     &    &    &   1 &    1 \\
 3  &   &    &     &     &     &     &     &     &     &    &  1 &   1 &    2 \\
 4  &   &    &     &     &     &     &     &     &     &  1 &  2 &   2 &    5 \\
 5  &   &    &     &     &     &     &     &     &   1 &  3 &  4 &   2 &   10 \\
 6  &   &    &     &     &     &     &     &   1 &   4 &  7 &  6 &   3 &   21 \\
 7  &   &    &     &     &     &     &   1 &   5 &  11 & 13 &  9 &   3 &   42 \\
 8  &   &    &     &     &     &   1 &   6 &  16 &  24 & 22 & 12 &   4 &   85 \\
 9  &   &    &     &     &   1 &   7 &  22 &  40 &  46 & 34 & 16 &   4 &  170 \\
 10 &   &    &     &   1 &   8 &  29 &  62 &  86 &  80 & 50 & 20 &   5 &  341 \\
 11 &   &    &   1 &   9 &  37 &  91 & 148 & 166 & 130 & 70 & 25 &   5 &  682 \\
 12 &   &  1 &  10 &  46 & 128 & 239 & 314 & 296 & 200 & 95 & 30 &   6 & 1365 \\
 13 & 1 & 11 &  56 & 174 & 367 & 553 & 610 & 496 & 295 & 125 & 36 &  6 & 2730 \\
 \bottomrule
\end{tabular}
\caption{\label{table:caterpillars}
The number $g_{s,n}$ of rooted binary perfect phylogenies with sample size $s$ and a caterpillar shape with $n$ leaves. Entries are obtained using Proposition~\ref{prop:cat rec}; the ``total'' is $g_s = \sum_{n=2}^s g_{s,n}$. The total follows A000975 (with the index shifted so that term $s$ in A000975 is $g_{s+1}$).}
 \end{table}

Because $g_{s,n}=0$ for $1 \leq s \leq n-1$, we add additional zeroes and simplify a convolution:
\begin{align*}
G_{n+1}(z) &=\sum_{s=1}^\infty \sum_{i=1}^{s-1} (g_{i,n} z^i)(b_{s-i,1} z^{s-i}) \\
&= G_{n}(z) \, B_1(z)\\
&=  \frac{z^n}{(1-z)^n(1+z)} \frac{z}{1-z}.
\end{align*}
The induction is complete.
\end{proof}

Note that $g_s$, the total number of caterpillar perfect phylogenies with sample size $s$, allowing all possible values of $n$, $2 \leq n \leq s$, follows the Lichtenberg sequence, OEIS sequence A000975 (the index is shifted, so that if A000975 is denoted $\{a_s\}$, then $a_s=g_{s+1}$). We verify this equivalence by showing an identity of generating functions. Denote the generating function for the number of caterpillars with sample size $s$, considering all possible numbers of leaves, by $G(z) = \sum_{s=2}^\infty g_s z^s$. 

We have that 
\begin{equation}
    \label{eq:gs}
g_s = \sum_{n=2}^s g_{s,n}
\end{equation}
and $G_n(z)=\sum_{s=n}^\infty g_{s,n} z^s$, and we add zeroes to obtain 
\begin{equation*} 
 G(z)= \sum_{s=2}^\infty g_s z^s = \sum_{s=2}^\infty \bigg(\sum_{n=2}^\infty g_{s,n}\bigg) z^s =  \sum_{n=2}^\infty \bigg(\sum_{s=2}^\infty g_{s,n} z^s \bigg)
= \sum_{n=2}^\infty G_n(z).
\end{equation*}
By Proposition~\ref{prop:cat_gf}, we then have
\begin{align} 
G(z) &=\sum_{n=2}^\infty \frac{z^n}{(1-z)^n (1+z)}\nonumber\\
&=\frac{1}{1+z} \bigg(\frac{1}{1-\frac{z}{1-z}}-\frac{z}{1-z}-1\bigg) \nonumber \\
&=\frac{z^2}{(1+z)(1-z)(1-2z)}, 
\label{eq: cat_total_gf}
 \end{align}
where the summation requires $|z| < \frac{1}{2}$. The Lichtenberg sequence has generating function $z/[(1+z)(1-z)(1-2z)]$, differing only in missing a factor of $z$, so that term $a_s$ accords with our $g_{s+1}$.

Note that if we were to consider the 1-leaf perfect phylogeny a caterpillar and to allow a trivial perfect phylogeny with $s=0$, then we would obtain a sequence $\{g'_s\}$ for the total number of perfect phylogenies with sample size $s$ and $n\geq 1$ leaves; for all $s \geq 0$, $g'_s=g_s+1$. This sequence, with generating function $G(z) + \frac{1}{1-z}$ to account for the extra perfect phylogeny with 1 leaf (for all $s \geq 1$) and the trivial perfect phylogeny ($s=0$), accords with A005578, which has generating function $(1-z-z^2)/[(1+z)(1-z)(1-2z)] = G(z)+\frac{1}{1-z}$.

\subsection{General unlabeled tree shape}\label{sec:gen_tree}

We can generalize the argument we have used for recursively counting perfect phylogenies with a caterpillar tree shape to an arbitrary tree shape with $n$ leaves. Let $T$ be an unlabeled tree shape with $|T|$ leaves. Tree $T$ has left and right subtrees, $T_\ell$ and $T_r$, with $|T_\ell|$ and $|T_r|$ leaves. In sequentially decomposing a tree into its left and right subtrees, eventually a single node or a cherry is reached, and the base case can be applied.

Denote by $N_{s,T}$ the number of rooted binary perfect phylogenies with unlabeled tree shape $T$, where $N_{s,T}=0$ if $s < |T|$ or $s$ is not an integer. 
\begin{proposition} 
\label{prop: unlabeled tree}
The number $N_{s,T}$ of rooted binary perfect phylogenies with unlabeled tree shape $T$ and sample size $s \geq |T|$ satisfies \\
(i) $N_{s,T}=1$ if $T$ has a single leaf and $s \geq 1$. \\
(ii) For $(s,T)$ with $s \geq |T| \geq 2$,
\begin{equation}
N_{s,T}=
\begin{cases}
\sum\limits_{i=|T_{\ell}|}^{s-|T_r|} N_{i,T_{\ell}}  N_{s-i,T_r},  & T_{\ell}\neq T_r,\\
\bigg(\sum\limits_{i=|T_{\ell}|}^{s-|T_r|} \frac{1}{2}N_{i,T_{\ell}}  N_{s-i,T_r} 
\bigg) + \frac{1}{2} N_{s/2, T_{\ell}}, & T_{\ell}=T_r.
\end{cases}
\label{eq:general tree recursion}
\end{equation}
\end{proposition}
\begin{proof}
(i) We have discussed the base case of the single-leaf tree in Proposition~\ref{prop:b_s,n recursion}i. (ii) For the general case, a perfect phylogeny with unlabeled shape $T$ is constructed from perfect phylogenies with unlabeled shapes $T_\ell$ and $T_r$. The minimal sample size assigned to $T_\ell$ is $|T_\ell|$, and the minimal sample size assigned to $T_r$ is $|T_r|$, so that the maximal sample size for $T_\ell$ is $s-|T_r|$. 

If $T_\ell \neq T_r$, then we sum the product of the number of perfect phylogenies for $T_\ell$ and the number of perfect phylogenies for $T_r$ over all possible values $i$ of the sample size assigned to $T_\ell$. Because $i \geq |T_\ell|$ and $s-i \geq |T_r|$, we have $i \leq s-|T_r|$.

If $T_\ell = T_r$ and sample size $i \neq \frac{s}{2}$ is assigned to the left subtree, then a factor of $\frac{1}{2}$ accounts for the fact that each perfect phylogeny traversed is also obtained for sample size $s-i$ assigned to the left subtree. If $T_\ell = T_r$ and $i=\frac{s}{2}$ (for even $s$), then we count the $\binom{N_{s/2, T_{\ell}}}{2}$ trees with distinct subtrees and the $N_{s/2,T_{\ell}}$  trees with identical subtrees. 
\end{proof}

For example, suppose $T$ is the 4-leaf symmetric unlabeled tree shape, with $T_\ell$ and $T_r$ both corresponding to the 2-leaf caterpillar. Proposition~\ref{prop: unlabeled tree} yields, for $s \geq 4$,
\begin{align}
N_{s,T} &=\bigg( \sum\limits_{i=2}^{s-2} \frac{1}{2}g_{i,2}g_{s-i,2} \bigg) +\frac{1}{2}g_{s/2,2} \nonumber \\ 
&=\bigg( \sum\limits_{i=2}^{s-2} \frac{1}{2}\bigfloor{ \frac{i}{2}} \bigfloor{ \frac{s-i}{2}} \bigg) + \frac{1}{2}\bigfloor{ \frac{s}{4} } \mathbbm{1}_{\{s \text{ is even}\}} \nonumber \\
&=\begin{cases}
\frac{(s + 1)(s - 1)(s - 3)}{48}, & \textnormal{odd } s \geq 5, \\
\frac{s(s - 1)(s -2)}{48} + \frac{1}{2}\floor{ \frac{s}{4}}, & \textnormal{even } s \geq 4.
\label{eq: noncat4}
\end{cases}
\end{align}
Note that the derivation follows the proof of Corollary~\ref{cor: s4}, eqs.~\eqref{eq:odd2} and \eqref{eq:even2}.

Recall that there are only two 4-leaf unlabeled topologies, the symmetric shape and the caterpillar. Adding eq.~\eqref{eq: noncat4}, counting perfect phylogenies for the symmetric shape, and eq.~\eqref{eq: C4}, counting those for the caterpillar, we obtain eq.~\eqref{eq:s4}, counting all perfect phylogenies with $n=4$ leaves. In particular, for odd $s$, using Lemma~\ref{lemma:evenodd} and Proposition~\ref{prop:gsn} to remove the floor function,
\begin{align}
\frac{(s+1)(s-1)(s-3)}{48}+\bigfloor{ \frac{s(s-2)(2s-5)}{24}} =&\frac{(s+1)(s-1)(s-3)}{48}+\frac{s(s-2)(2s-5)}{24}-\frac{1}{8} \nonumber \\
=& \frac{(s-1)(s-3)(5s-1)}{48}. \nonumber 
\end{align}
For even $s$, by Proposition~\ref{prop:gsn},
\begin{align}
\frac{s(s-1)(s-2)}{48}+\frac{1}{2}\bigfloor{ \frac{s}{4}}+ \bigfloor{ \frac{s(s-2)(2s-5)}{24}} 
 = &\frac{s(s-1)(s-2)}{48}+\frac{1}{2}\bigfloor{ \frac{s}{4}}+\frac{s(s-2)(2s-5)}{24} \nonumber \\
=& \frac{s(s-2)(5s-11)}{48}+\frac{1}{2}\bigfloor{ \frac{s}{4}}. \nonumber
\end{align}

\section{Discussion}

We have studied the enumerative combinatorics of rooted binary perfect phylogenies. We have provided a recursive formula to enumerate the rooted binary perfect phylogenies with a given sample size $s$ via eq.~\eqref{eq:enumeration_alln}, and we have provided an asymptotic approximation in eq.~\eqref{eq:approx}. We have also refined the enumeration, counting rooted binary perfect phylogenies for a given sample size $s$ separately for each possible value of the number of leaves $n$ via eq.~\eqref{eq: binary_rec}. We have counted rooted binary perfect phylogenies associated with specific shapes (eq.~\eqref{eq:general tree recursion}), notably a caterpillar shape (eq.~\eqref{eq:caterpillarcounts}). A summary of results appears in Table~\ref{table:results}.

The enumerations build on the efforts of Palacios et al.~\cite{palacios2022} to enumerate the labeled and unlabeled topologies and labeled and unlabeled histories that can be associated with a rooted perfect phylogeny (binary or multifurcating). For rooted binary perfect phylogenies, we provide enumerations that can be employed as a starting point for the enumerations of labeled and unlabeled topologies and labeled and unlabeled histories by Palacios et al.~\cite{palacios2022}.

The recurrence for $b_s$, the number of rooted binary perfect phylogenies with sample size $s$ (eq.~\eqref{eq:enumeration_alln}), is similar to the recurrence for the number of rooted binary unlabeled trees (eq.~\eqref{eq:wedderburn_rec})---except that it requires the addition of a 1 for the single-leaf perfect phylogeny whereas the recurrence for the rooted binary unlabeled trees does not include a corresponding possibility. This small difference leads to a large difference in asymptotic growth. Whereas the asymptotic growth of the rooted binary unlabeled trees---the perfect phylogenies with sample size $s$ and $s$ leaves---is approximately $0.3188(2.4833)^s s^{-3/2}$, the growth of the rooted binary perfect phylogenies with sample size $s$ is substantially larger, approximately $0.3519(3.2599)^s s^{-3/2}$ (eq.~\eqref{eq:approx}). 

Some of our results produce known integer sequences. The sequences for $b_s$ (OEIS A113822) and $b_{s+1,s}$ (OEIS A085748) have been reported but little studied; $b_{s+1,s}$ counts rooted binary labeled trees with $s$ leaves in which all leaves except one are labeled ``1'' and the last leaf is labeled ``2''; equivalently, it is the number of rooted binary trees that are unlabeled except that one leaf is given a label. Sequence $b_{s,3}$ follows OEIS A002620 (Corollary~\ref{cor:s3}), the ``quarter-squares.'' The number of ways to place a given sample size across \emph{some} caterpillar shape(Table~\ref{table:caterpillars}) follows OEIS A000975, a sequence well studied in other contexts.

Perfect phylogenies have applications in multiple biological settings. Palacios et al.~\cite{palacios2022} motivated their study from an interest in inference of evolutionary trees. Perfect phylogenies appear in problems concerning DNA sequences descended in a population from an ancestral sequence by a process with little or no genetic recombination~\cite[pp.~460-462]{Bafna2004, Clark06, Gusfield14}. Recently, they have been considered in problems with cell lineages and tumors~\cite{ElKebir2016, Wu2020}. Our enumerative results assist in characterizing the sizes of sets of perfect phylogenies relevant to the various biological applications. 

\begin{table}[tb]
    \centering
    \begin{tabular}{l|lll}
    \hline
Tree shapes & Recursion & Generating function & Asymptotics\\
\hline
All shapes           & $b_s$,     Proposition~\ref{prop:bs recursion}    & $B(z)$,   Proposition~\ref{prop:Bs gen func} & Theorem~\ref{thm: binary asymptotics} \\
All $n$-leaf shapes  & $b_{s,n}$, Proposition~\ref{prop:b_s,n recursion} & $B_n(z)$, Proposition~\ref{lemma: bn gf}     & - \\
$n$-leaf caterpillar & $g_{s,n}$, Proposition~\ref{prop:cat rec}         & $G_n(z)$, Proposition~\ref{prop:cat_gf}      & - \\
All caterpillars     & $g_s$,     Equation~\ref{eq:gs}                   & $G(z)$,   Equation~\ref{eq: cat_total_gf}    & - \\
Arbitrary shape      & $N_{s,T}$, Proposition~\ref{prop: unlabeled tree} & -                                            & - \\
\hline
 \end{tabular}
    \caption{The main results of the paper. We have variously obtained recursions, generating functions, and asymptotics for the number of rooted binary perfect phylogenies with sample size $s$: considering all tree shapes, all tree shapes with $n$ leaves, the $n$-leaf caterpillar shape, all caterpillar shapes, and a single shape that is specified, but that is arbitrary.}
    \label{table:results}
\end{table}

From the perspective of the lattice formulation for perfect phylogenies (Figure \ref{fig:2lattice5}), we have evaluated the number of non-empty elements of the lattice, $b_s$, and $b_{s,n}$, the number of elements that lie $s-n+1$ ``steps'' from the minimal element $\phi$ to the maximal element, a single leaf. However, in describing lattices of binary perfect phylogenies, we have left a number of questions unanswered. In how many ways can the lattice be traversed---via the order relation---between the minimal and maximal perfect phylogenies? How many perfect phylogenies exist with specified features, perhaps concerning numbers of nodes with different numbers of descendants or leaves with specified multiplicities? Applications of the lattice formulation may provide further insights on these questions.

\section*{Acknowledgements}
We thank Michael Fuchs, Bernhard Gittenberger, and Julia Palacios for discussions. Support was provided by a National Science Foundation Graduate Research Fellowship and by National Institutes of Health grant R01 HG005855.

\bibliography{references}


\end{document}